\documentclass[a4paper,11pt]{article}
\usepackage{amsmath}
\usepackage[retainorgcmds]{IEEEtrantools}
\usepackage{amsthm}
\usepackage{amsfonts}
\usepackage{amsrefs}

\usepackage{url}
\usepackage{graphicx}
\usepackage{afterpage}
\usepackage{placeins}

\newtheorem{Thm}{Theorem}[section]

\newtheorem{Ass}[Thm]{Assumption}
     
\newtheorem{Def}[Thm]{Definition}

\newtheorem{Prop}[Thm]{Proposition}

\newtheorem{Lem}[Thm]{Lemma}

\newtheorem{Cor}[Thm]{Corollary}

       \title{Pricing and Hedging the No-Negative-Equity Guarantee in Equity-Release Mortgages}
 \usepackage{authblk}

\author[1]{Kevin Engelbrecht\thanks{E-mail: \textit{k.engelbrecht@warwick.ac.uk}}}
\author[1,2]{Saul Jacka\thanks{Corresponding author. Saul Jacka gratefully acknowledges funding received from the EPSRC grant EP/P00377X/1 and is also grateful to the Alan Turing Institute for their financial support under the EPSRC grant EP/N510129/1.}
\thanks{ E-mail: \textit{s.d.jacka@warwick.ac.uk}\\Department of Statistics\\University of Warwick\\Coventry CV4 7AL, UK}}
\affil[1]{University of Warwick}
\affil[2]{Alan Turing Institute}

\begin{document}
\maketitle
  
\begin{abstract}
	We provide a practical superhedging strategy for the pricing and hedging of the No-Negative-Equity-Guarantee (NNEG) found in Equity-Release Mortgages (ERMs), or reverse mortgages, using a discrete-time model. In contrast to many papers on the NNEG and industry practice we work in an incomplete market setting so that deaths and property prices are not independent under most pricing measures. We give theoretical results and numerical illustrations to show that the assumption of market completeness leads to a considerable undervaluation of the NNEG. By introducing an Excess-of-Loss reinsurance asset, we show that it is possible  to reduce the cost of the  superhedge for a portfolio of ERMs with the average cost decreasing rapidly as the number of lives in the portfolio increases. All the hedging assets, with the exception of cash, have a term of one year making the availability of a property hedging asset from over-the-counter derivative providers more realistic. We outline how a practical multi-period ERM pricing and hedging model can be built. Although the prices identified by this model will be higher than prices under the completeness assumption, they are considerably lower than those under the Equivalent Value Test mandated by the UK's Prudential Regulatory Authority.
	\end{abstract}
	


      	\section{Introduction}
	Equity-release mortgages(ERMs) or reverse mortgages have become very popular in recent years, both with income-poor capital-rich homeowners wishing to release some equity in their properties and with insurers seeking higher-yielding assets to back  	their annuity portfolios as yields on gilts and bonds have dropped to historic lows (according to the UK's Equity Release Council, insurers issued \pounds 3.92 billion of ERM loans in 2019).
	
Given the relatively large size of ERM assets backing pensioners' annuities and the complex derivative embedded in these contracts, concerns have been raised over the valuation approach adopted by insurers (see \cite{dowd2018asleep}, \cite{dowd2019valuation} and \cite{IrishERM}). 

Under a typical ERM policy, the homeowner  receives a loan at the start of the contract. The loan increases with interest until the homeowner dies or sells the property at time $t$.  Almost all ERM policies have what is called a "No Negative Equity Guarantee" (NNEG) which means that the homeowner is not liable for the shortfall 
\begin{equation}
	(L_t-S_t)^+ \label{eq:PropPutPayout}
\end{equation}
between the accumulated loan amount $L_t$ and the sale proceeds $S_t$. 

We assume initially that $L_t$ is the same for all policies in the portfolio (our results are generalised in subsection \ref{VaryingLoans} to the case where loan amounts differ). This means that if there are $D_t$ deaths in year $t$ in a portfolio of ERM policies, the insurer would lose cashflows to the value of 
\begin{equation*}
	 	D_tL_t-D_t\min(S_t,L_t)=D_tL_t-D_t(L_t-(L_t-S_t)^+)=D_t(L_t-S_t)^+.
	 \end{equation*}
This shows that all lifetime mortgages with the NNEG contract condition embody a derivative with the payout
	 \begin{equation}
	 	D_t(L_t-S_t)^+.	\label{eq:NNEGPayout}
	 \end{equation}	 
 We refer to the option with the payout (\ref{eq:NNEGPayout}) as the NNEG option (or NNEG, for short) and the embedded property put option with payout (\ref{eq:PropPutPayout}) as the property put.

At the time of writing, UK insurance companies typically value ERMs by using a discounted (expected) cashflow (DCF) approach. The  cost of the NNEG option for each time period $t$ is calculated as the expected number of deaths times the cost of the property put (usually calculated using the Black-Scholes formula). For a risk-free discount rate of zero and a pricing measure $\mathbb{Q}$ with the correct marginals, this is equivalent to

\begin{equation}
	\text{Cost NNEG } = \mathbb{E^Q}[D_t]\mathbb{E^Q}[(L_t-S_t)^+].		\label{eq:DCFValue}
\end{equation}

We believe that there is a significant problem with this valuation approach because, by (\ref{eq:NNEGPayout}),  we should rather be calculating $\mathbb{E^Q}[D_t(L_t-S_t)^+]$. If $D_t$ and $S_t$ are independent under $\mathbb{Q}$ (for example, if $\mathbb{Q}$ is the independence coupling of $D_t$ and $S_t$), then we do have

\begin{equation}
	\mathbb{E^Q}[D_t(L_t-S_t)^+]=\mathbb{E^Q}[D_t]\mathbb{E^Q}[(L_t-S_t)^+].	\label{eq:IndepDeathsAndHousePrice}
\end{equation}

Unfortunately, because the insurance markets are incomplete, there are many pricing measures or equivalent martingale measures (EMM) $\mathbb{Q}$. Under most of these EMMs $\mathbb{Q}$, (\ref{eq:IndepDeathsAndHousePrice}) does not hold. According to asset pricing theory (see \cite{MathArbit2006} or \cite{follmer2011stochastic}  for an introduction in discrete time), there will be a range of arbitrage-free prices $[P_L,P_H]$, where 

\begin{equation}
	P_L:= \inf_{\text{EMMs }\mathbb{Q}}\mathbb{E^Q}[D_t(L_t-S_t)^+] 
\end{equation}
and
\begin{equation}
	P_H:= \sup_{\text{EMMs }\mathbb{Q}}\mathbb{E^Q}[D_t(L_t-S_t)^+]. 							\label{eq:PH}
\end{equation}

The price of the NNEG option calculated using the DCF method will lie in the range $[P_L,P_H]$ and is unlikely to be sufficiently large to construct a hedge that will cover all of the claims under the NNEG option. 

We note that there is a promising result in the paper of  Jacka et al. (\textsection 2.3 in  \cite{jacka2015coupling}) in which they show that a continuous-time $(\mathcal {F}_t)$-Markov chain  and $(\mathcal {F}_t)$-Brownian motion under a common filtration $(\mathcal {F}_t)$ are necessarily independent. This would appear to justify the assumption of independence between the deaths (Markov) process  and the property-price process for continuous-time models of the ERM. However, in {\em practice} there are many reasons that could invalidate the independence assumption (for example, ERMs are  discretely-valued and hedged  rather than continuously-hedged and trading costs are present, to name but two).

Without a coherent hedging strategy to handle the complex nature of the NNEG option, over-the-counter derivative providers are less likely to provide insurers with hedging assets. 

The main contribution of this paper is that we find a superhedging strategy for the NNEG option which turns out to be remarkably cheap. This is surprising because superhedges are usually prohibitively expensive to set up (but on the positive side superhedges are very prudent) and are not usually regarded as feasible hedging strategies.

 We do so by assuming the availability of a realistic additional hedging asset--- an excess of loss reinsurance (XoL)  contract --- and use it to price a portfolio of $n$ ERMs simultaneously. 

The motivation for this approach is the following: for an ERM portfolio consisting of infinitely-many lives the Strong Law of Large Numbers (SLLN) would allow us to guarantee the proportion $p$ of lives that die  and therefore we could perfectly hedge the NNEG at a cost of $pq$ per ERM by purchasing $p$ property put options, where $q$ is the cost of the property put. 

By assuming that there is a reinsurance contract available that pays out the excess number of deaths above 
\begin{equation}
	e(n)=np(1+\epsilon),\qquad n \in \mathbb{N}, \label{eq:e(n)}
\end{equation}
where $\epsilon > 0$, we can price the portfolio of ERMs as though the portfolio had {\em close to infinitely many lives}.

Indeed, in Theorem \ref{SuperhedgePrices}, we show that a cheapest superhedge for the NNEG option will be one of the following superhedges constructed from the candidate hedging assets (cash, a property stock, the XoL contract and  life assurances on each of the $n$ lives): 
\begin{itemize}
\item a single Group Life Assurance (GLA) costing $np$
\item $n$ property puts costing $nq$
\item $e(n)$ property puts and a single XoL contract with excess $e(n)$ costing $e(n)q+X_0^e(n)$, where $X_0^e(n)$ is the price of the XoL contract
\end{itemize}

and, then in Proposition \ref{AsymptoticHedge}, we show that, for  $n$ sufficiently large and $\epsilon>0$ sufficiently small, the cheapest superhedge is always $e(n)$ property puts and an XoL contract with excess $e(n)$.
The average cost per life of this superhedge is
\begin{equation*}
	\frac{1}{n}(e(n)q+X_0^{e(n)})=pq(1+\epsilon)+\frac{1}{n}X_0^{e(n)}\rightarrow pq(1+\epsilon),\text{as }n\rightarrow \infty.
\end{equation*}

In effect, the price of the reinsurance contract is the price to be paid to gain access to the SLLN. Notice that, by  (\ref{eq:DCFValue}),  the "DCF Black-Scholes" approach gives a value of only $pq$ to the NNEG which is the same value as the  the hypothetical portfolio consisting of infinitely-many lives mentioned earlier.

The use of an excess slightly larger than the average number of deaths reflects the reality that reinsurance would not be available without some such "experience margin". A very useful financial consequence is that  a Large Deviations Principle will apply to the price of the reinsurance contract and its cost per life will be very small for even relatively small portfolios of lives. In fact, in Proposition \ref{XoL_LDP}, for the single-period model, we show that the price of the reinsurance contract tends to zero exponentially fast as the number of lives goes to infinity.

All the assets, with the exception of cash, will have a term of one period (usually a year) only. One reason for this is that we shall be using a static mark-to-market hedge-and-forget strategy. At the end of each year the superhedge will cover all claims occurring during the year plus a potential release of surplus. The use of short-term disposable assets for hedging makes sense in practice: over-the-counter derivative providers are more willing to provide property derivatives linked  to residential property or property indices if the term of the contract is short (see summary of responses to Question 8 in CP48/16  
\cite{CP48/16})

The single-period superhedge can be easily extended to multiple time periods. When the first time period terminates, any surplus released from the superhedge contributes to the setup costs of the superhedge for the next period. In order to ensure that there is sufficient release of surplus to construct the next superhedge, the hedging strategy should be calculated working backwards in time. For example, for a model with $T$ time periods and $N$ lives, a multinomial tree model can be constructed with $(t+1)(n+1)$ nodes at time $t\in \{1,2,\cdots,T\}$. Starting at time $T-1$, the superhedge for each time $(T-1)$-node is given by Theorem \ref{SuperhedgePrices}.  Using the superhedging costs determined in  the previous time step, the superhedges for the time $(T-2)$-nodes can then be determined using linear programming. Continuing in this way  the cheapest time-zero superhedge can be found. 

The number of nodes in the tree and hence the number of linear optimisation problems that need to be solved by the computer is not excessive. For example, for $T=40$ years with annual time intervals and $N=100$ lives, there would only be  $40\times101=4040$ nodes at time $T-1$.

\section{Numerical results for the NNEG option}
\subsection{Single-period results}

The graph in figure \ref{fig:SinglePeriod} shows how the average cost of the cheapest superhedge reduces with increasing number of lives $n$. 

The probability of death $p$ is $0.45$ and the cost of the property put $q$ is $0.5374$. The XoL price $X_0^{e(n)}$ is calculated assuming that the number of deaths is Bernoulli-distributed with parameters $n,p$ and $e(n)$ is given by (\ref{eq:e(n)}) with $\epsilon=0.1$. 

We can see from the graph that, for small $n$, the cheapest superhedge costs $0.45$ which corresponds to holding a GLA. Note that for these small values of  $n$, $e(n)<1$ and the XoL contract reduces to a GLA on $n$ lives and that $p<q$. For larger $n$, the superhedging strategy of holding one XoL contract with excess $e(n)$ and $e(n)$ property puts becomes cheaper as predicted by Theorem \ref{SuperhedgePrices} and Proposition \ref{AsymptoticHedge}.

\begin{figure}
   \includegraphics[width=4in]{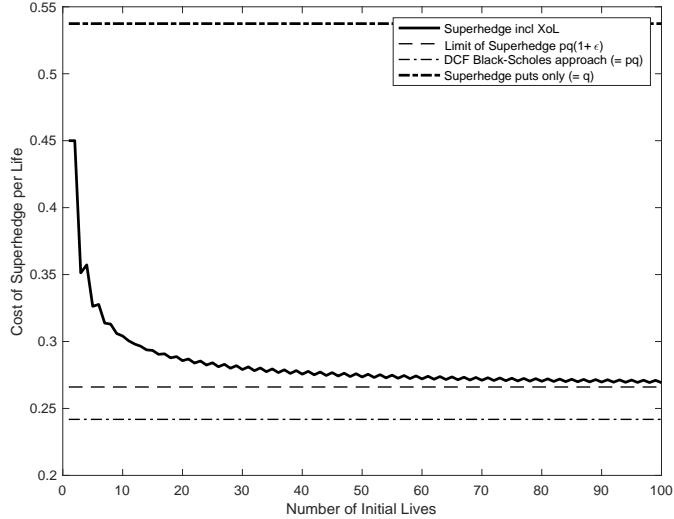}
     \caption{Single-period model - Superhedge cost versus number of initial lives. Notes: probability of death $p=45\%$, initial house price $=100$, probability house price falls $q=53.74\%$, initial loan $=87.07$, $\epsilon=10\%$, risk-free interest rate $=0\%$}
      \label{fig:SinglePeriod}
   \end{figure}

\afterpage{\FloatBarrier}

\subsection{Multi-period results}

We built a practical multi-period model in Matlab which can calculate the superhedging cost and strategy for the NNEG using the method described in the introduction. We chose parameters similar to those in Table 1 in  \cite{dowd2019valuation} for comparison purposes. Our parameters were the following:

\begin{itemize}
\item the age of policyholders at the start of the ERM contract is 70
\item the initial house price is 100
\item house price volatility is 15\% p.a.
\item the initial loan amount is 40 with loan interest of 5\% p.a.
\item the risk-free rate is 0\% p.a. 
\item the mortality table used is A67/70 with no adjustments for early redemptions or long term care exits
\item the dividend yield (deferment rate) for the property stock is assumed to be zero 
\end{itemize}

The results of the multi-period model for varying number of lives is shown in  figure \ref{fig:MultiPeriod}. The value "DCF Black-Scholes" in the figure is the value of the NNEG determined using (\ref{eq:IndepDeathsAndHousePrice}) for multiple periods. i.e. it is equal to $\sum_1^Tq_tP(0,L_t)$, where $q_t$ is the expected number of deaths in period $t$ and $P(0,L_t)$ is the price of a European put option at time zero, with term $t$ and  strike $L_t$  equal to the accumulated loan amount at time t. The ERM actuaries often call this the risk-neutral Black-Scholes NNEG value. As can be seen from figure \ref{fig:MultiPeriod}, it is significantly lower than the superhedge even for a large number of lives $n$. 
\subsection{Impact of deferment rates}
A modification of the risk-neutral Black-Scholes approach used by some insurers is the so called "real-world Black-Scholes" valuation approach, where house prices are inflated at a rate in excess of the risk-free interest rate, producing an even lower value for the NNEG than the risk-neutral Black-Scholes approach! 

Because insurers were/are using these low valuations for the NNEG in  the net cashflow calculation for the construction of securitised assets to back their annuity portfolios, the Equivalent Value Test (EVT) was introduced by the UK's Prudential Regulatory Authority (PRA) as a \lq\lq diagnostic tool" to help insurers assess  whether the yields on their securitised securities is excessive (see the UK's PRA publication SS2/17 on illiquid assets \cite{SS2/17}). For the EVT test, the PRA chose the Black-Scholes model for a dividend-paying stock to value the NNEG using the concept of "deferment rates" to justify the use of such a model. This model can produce values for the NNEG which are considerably higher than those of the standard Black-Scholes non-dividend-paying stock model for reasonable dividend rates. A dividend(deferment) rate of $q=4.2\%$ was used to produce the values in Table 1 in  \cite{dowd2019valuation}. The lowest NNEG value quoted there was 31.5 which is between two and three times higher than the "Superhedge including XoL"  NNEG cost in figure \ref{fig:MultiPeriod}.  

\begin{figure}
    \includegraphics[width=4in]{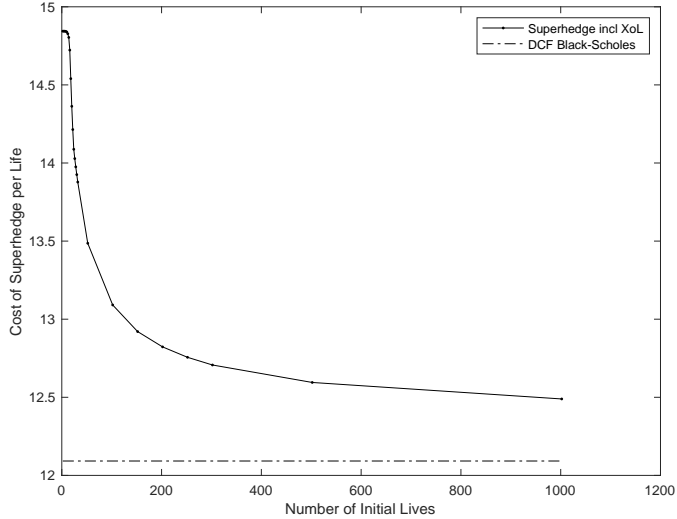}
    \caption{Multi-period model - Superhedge cost versus number of initial lives. Notes: Age $=70$, initial house price $=100$, house price volatility $=15\%p.a.$, initial loan $=40$, loan interest $=5\%p.a.$, risk-free interest rate $=0\%p.a.$}
    \label{fig:MultiPeriod}
    \end{figure}

\afterpage{\FloatBarrier}

\section{Setting}
	We represent time by $t\in\{0,1\}$. At time $t=0$, we assume that there are $n$ identical lives, each of whom purchases a lifetime mortgage.  The financial market consists of an insurance market with assets which are insurance contracts written on the $n$ lives and a non-insurance market consisting of a cash bond and a property stock. We assume that all assets are liquid and can be traded frictionlessly. Any long or short position may be held in the assets. We  refer to the combined insurance and non-insurance market as the Combined Market.
	
To simplify the presentation, we assume that the Combined Market is normalised. i.e. all price processes have been discounted by the cash bond price process. The superhedging strategies or pricing measures are not affected by working in a normalised market. 

If one wishes to explicitly model the effects of interest, then one can multiply the relevant formulae by the factor $1+r$, where $r$ is the assumed rate of interest earned by the original cash bond. For example, the undiscounted version of the NNEG payout (\ref{eq:NNEGPayout}) is ${D_t(L_t(1+r)-S_t(1+r))^+}$. The pricing constraints (\ref{eq:ACMMStock}) - (\ref{eq:ACMMXoL}) are unaffected because we simultaneously inflate the asset payouts  and divide by $(1+r)$ for discounting. For a similar reason, there is no net effect on duality equation (\ref{eq:DualityEquality}).

	\subsection{The probabilistic framework}
	
	 The Combined Market is  represented by the measurable space $(\Omega,\mathcal{F})$, where $\Omega=\{0,1\}^{n+1}$ and $\mathcal{F}=\mathcal{P}(\Omega)$ is the powerset of $\Omega$. 
	 
	 Corresponding values are as follows:
	 
	 For $\omega=(\omega_{0},\cdots,\omega_{n})\in\Omega$, $S_1=u1_{(\omega_0=1)}+d1_{(\omega_0=0)}$ so that  $\omega_{0}$ is the indicator for an up-jump in the property price index. 
	 Similarly if $D^i$ is the indicator that the $i$th life has died by time 1, then $D^i={\omega_i}$. The number of deaths is given by the random variable $D_1 = \omega_1+\cdots +\omega_n$. We sometimes write $D_1(n)$ instead of  $D_1$ if we wish  to emphasise the number of lives.

\subsection{The non-insurance market}
	The non-insurance market consists of the following assets:
	
	\begin{enumerate}
	\item A (normalised) cash bond with a constant price of one. 	
	\item A tradable property stock (which could be an asset provided by an over-the-counter provider such as a derivative on a property index or residential properties) whose value the price processes of the $n$ properties are assumed to follow. The model for the price process $\{S_{t}\}_{t=0,1}$ of the property stock is assumed to be a binomial tree
		\begin{equation*}
			S_{1}(\omega)=
			\begin{cases}
				S_{0}u, & \text{if } \omega_{0}=1 \\
				S_{0}d, & \text{otherwise,}  
			\end{cases}
		\end{equation*}
for $\omega=(\omega_{0},\cdots,\omega_{n})\in \Omega$, where $S_{0}>0$  is the price at time zero of the property stock. To ensure that the binomial stock model does not have arbitrage, we assume that $0<d<1<u$.	
	
	\end{enumerate}

\subsection{The insurance market}
	
	The insurance market consists of the following assets commencing at time $t=0$ and written on the $n$ 	lives:
	
	\begin{enumerate}
	
	\item $n$ single-life assurances written on each of the $n$ lives  for a price $p\in(0,1)$.  If the life dies over $[0,1]$, then the contract pays out an amount of one; otherwise the contract expires worthless. We denote its price process by $\{L^{(i)}_t\}_{t=0,1}$. So $L^{(i)}_0(\omega)=p$ and  $L^{(i)}_1(\omega)=\omega_i$, for $\omega\in\Omega$ and $i\in\{1,2,\cdots,n\}$.

	\item A group life insurance (GLA) written on all $n$ lives costing $np$. The GLA pays out an amount of one for each life that has died over the period $[0,1]$. The GLA has the price process $\{G_t\}_{t=0,1}$. The 	payout or 	price at time one of the GLA is
			\begin{equation}
				G_1(\omega)=D_1(\omega),\qquad \omega \in \Omega.
			\end{equation}	
			
Note that the GLA is a redundant asset because it can be constructed from a holding of $n$ single-life assurances, each costing $p$. However, we  see later that, for superhedging, only a GLA  and not individual single-life assurances is needed.
	
	\item An excess of loss reinsurance (XoL) on the $n$ lives with excess $e\in[0,n]$. We  denote its price process by $\{X_{t}^{e}\}_{t=0,1}$. The payout or price of the XoL at time one  is given in terms of the total number of deaths $D_1$ by
		\begin{equation}
			X_{1}^{e}(\omega)=(D_{1}(\omega)-e)^{+},\qquad \omega \in \Omega.
		\end{equation}
		If we wish to emphasize that $e$ and $X_0^e$ are functions of the number of lives $n$, then we  denote them 	by $e(n)$ and $X_0^e(n)$. If we wish  to emphasise the probability of death $p$ as well, we  write $X_0^e(n;p)$.
		
	\end{enumerate}

\subsection{The lifetime mortgage}
	At time $t=0$ the insurer issues $n$  identical lifetime mortgages  with loan amount $L>0$. The loan amount is assumed to be less than the value of the property. We assume that any interest added to the loan is included in the loan amount $L$. We do not consider any other contract terminations or redemptions such as long-term care (LTC) or downsizing. The hedging and pricing strategy developed in this paper could be easily extended to incorporate LTC exits if there are LTC insurance contracts available to use as hedging assets.

\subsection{The set of pricing measures}

We denote by $\mathbb{P}$ the reference measure on $(\Omega,\mathcal{F})$.

\begin{Ass}\label{ass:ProbP}
	Under $\mathbb{P}$, the random variables $\{L_1^{(i)}\}_{i=1}^n$ are independent and Bernoulli-distributed with probability of success(death) of $p$. 
\end{Ass}
Note that for simplicity we have assumed that lives are independent under the reference measure $\mathbb{P}$.  The probability measure $\mathbb{P}$ will be used as a reference measure against which all pricing measures must be absolutely-continuous. We work with the larger class of absolutely-continuous pricing measures rather than equivalent pricing measures because the extremal pricing measures for (\ref{eq:PH}) which we find in Theorem \ref{SuperhedgePrices} are only absolutely-continuous.

Note that the superhedging price is not increased if we take the supremum in (\ref{eq:PH}) over the larger set of all absolutely-continuous pricing measures with respect to $\mathbb{P}$  (see Remark 1.36 in \cite{follmer2011stochastic}).

In order to calculate the superhedging price of the random payouts (\ref{eq:NNEGPayout}) in the Combined Market, we need to determine $\cal M$, the collection of probability measures $\mathbb{Q}$ on $(\Omega, \cal F)$ which are absolutely-continuous with respect to $\mathbb{P}$ and such that all discounted asset price processes are martingales under $\mathbb{Q}$.  
\begin{Def}
 Members of $\mathcal{M}$ are called absolutely-continuous  martingale measures (ACMMs).

\end{Def}

The constraints imposed on $\mathbb{Q}$ to belong to $\cal M$ are:

\begin{enumerate}
	\item $\mathbb{Q}$ must satisfy
	\begin{equation*}
		S_0=S_0d\mathbb{Q}(S_1=S_0d)+S_0u\mathbb{Q}(S_1=S_0u).
	\end{equation*}
	After simplifying, we obtain the standard Binomial tree \emph{stock pricing constraint}
	\begin{equation}
		\mathbb{Q}(S_1=S_0d) =q:=\frac{u-1}{u-d}. \label{eq:ACMMStock}
	\end{equation}
	
	\item  $\mathbb{Q}$ must satisfy
	the \emph{single-life assurance pricing constraint}
	\begin{equation}
		\mathbb{Q}(L_1^{(i)}=1) =p. \label{eq:ACMMSLA}
	\end{equation}
	
	\item $\mathbb{Q}$ must satisfy the \emph{XoL pricing constraint}:	
	\begin{equation}
		X_0^e=\sum_{k>e}^n(k-e)\mathbb{Q}(D_1=k). \label{eq:ACMMXoL}
	\end{equation}
\end{enumerate}

The pricing constraints (\ref{eq:ACMMStock}) - (\ref{eq:ACMMXoL}) suggest, that for pricing the NNEG option, we can replace the single life pricing constraint (\ref{eq:ACMMSLA}) by the weaker GLA constraint
	\begin{equation}
		G_0=np=\sum_{k=1}^nk\mathbb{Q}(D_1=k) \label{eq:ACMMGLA}
	\end{equation}
 and use a subset of $\mathcal{M}$ of exchangeable measures defined below:

\begin{Def}
The set $\mathcal{M}^e\subset\mathcal{M}$ of  exchangeable measures is defined to be the set of all $\mathbb{Q}\in \mathcal{M}$ under which the sequence of random variables $\omega_1, \cdots, \omega_n$ is exchangeable for $\omega=(\omega_0,\omega_1, \cdots ,\omega_n) \in \Omega$.
\end{Def}
Note that, in the definition of exchangeability, the first element $\omega_0$ representing property prices is not included in the exchangeable sequence. 

\begin{Prop} \label{Prop:QdashwithMarginals}
	Given $x_k,y_k \ge 0$ and $z_k=x_k+y_k$, for $k=0,1,\cdots ,n$ satisfying
\begin{IEEEeqnarray}{rCl}
	1 		&=& 		\sum_{k=0}^n z_k, \\
	q		&=& 		\sum_{k=0}^n y_k,\\
	np		&=&		\sum_{k=0}^n kz_k,		\label{eq:npsum} \\
	X_0^e	&=&		\sum_{k=0}^n (k-e)^+z_k,
\end{IEEEeqnarray}

there exists $\mathbb{Q'} \in \mathcal{M}^e$ such that
\begin{IEEEeqnarray}{rCl}
	x_k &=& \mathbb{Q'}(D_1=k,S_1=S_0u)\\
	y_k &=& \mathbb{Q'}(D_1=k,S_1=S_0d),
\end{IEEEeqnarray}
for $k=0,1,\cdots,n$.
\end{Prop}

\begin{proof}
	Define $\mathbb{Q'}$ as follows: for $\omega = (\omega_0, \omega_1, \cdots,\omega_n) \in \Omega$ with $k=\sum_1^n\omega_i$, let
	\begin{IEEEeqnarray}{rCl}
	\mathbb{Q'} ({(0,\omega_1,\cdots,\omega_n)}) &=& \frac{1}{\binom{n}{k}}y_k, \\
	\mathbb{Q'} ({(1,\omega_1,\cdots,\omega_n)}) &=& \frac{1}{\binom{n}{k}}x_k.
\end{IEEEeqnarray}
Then
\begin{IEEEeqnarray}{rCl}
	\mathbb{Q'} (D_1=k,S_1=S_0d) &=& y_k, \label{eq:S0d}  \\
	\mathbb{Q'} (D_1=k,S_1=S_0u) &=& x_k. \label{eq:S0u}
\end{IEEEeqnarray}

From its construction, it is clear that $\mathbb{Q'}$ is an exchangeable measure on $\mathcal{F}$. We  prove that $\mathbb{Q'}$ is a martingale measure by showing that it satisfies (\ref{eq:ACMMStock}) - (\ref{eq:ACMMXoL}). For (\ref{eq:ACMMStock}),
\begin{equation*}
	\mathbb{Q'}(S_1=S_0d) = \sum_{k=0}^n \mathbb{Q'}(D_1=k,S_1=S_0d)=\sum_{k=0}^n y_k = q.
\end{equation*}

Note that, from (\ref{eq:S0d}) and  (\ref{eq:S0u}), it follows that $\mathbb{Q'} (D_1=k) = z_k$. Then for (\ref{eq:ACMMSLA}),
\begin{IEEEeqnarray*}{rCl}
		\mathbb{Q'}(i^{\text{th}}  \text{ life dies}) & = & \sum_{k=1}^n \mathbb{Q'}(i^{\text{th}}  \text{ life dies and }D_1=k) \\
		& = & \sum_{k=1}^n \mathbb{Q'}(i^{\text{th}}  \text{ life dies}|D_1=k)\mathbb{Q'}(D_1=k) \\
		& = & \sum_{k=1}^n \frac{\binom{n-1}{k-1}} {\binom{n}{k}} \mathbb{Q'}(D_1=k) \text{,by exchangeability} \\
		& = &  \frac{1} {n}\sum_{k=1}^n kz_k           \\
		& = & \frac{1}{n} (np), 	\qquad \text{by } (\ref{eq:npsum})\\
		& = & p.
\end{IEEEeqnarray*}

For (\ref{eq:ACMMXoL}),
\begin{IEEEeqnarray*}{rCl}
	\sum_{k=0}^n (k-e)^+\mathbb{Q'}(D_1=k) = \sum_{k=0}^n (k-e)^+z_k = X_0^e.
\end{IEEEeqnarray*}

It follows that $\mathbb{Q'}$ is an ACMM and $\mathbb{Q'} \in \mathcal{M}^e$.
\end{proof}

\begin{Cor}\label{EMeqEMdash}
	Let $\mathcal{G} = \sigma (D_1,S_1)$. Given any $\mathbb{Q} \in \mathcal{M}$, there exists $\mathbb{Q'} \in \mathcal{M}^e$ such that $\mathbb{Q}  |_{\mathcal{G}} = \mathbb{Q'}  |_{\mathcal{G}}$
so that
	\begin{equation}
		\sup_{\mathbb{Q} \in \mathcal{M}} \mathbb{E^Q} [D_1(L-S_1)^+] = \sup_{\mathbb{Q'} \in \mathcal{M}^e} \mathbb{E^{Q'}} [D_1(L-S_1)^+]. 			\label{eq:supEQeqEQdash}
	\end{equation}

\end{Cor}

\begin{proof}
	For $k=0,1,\cdots ,n$, let
	\begin{IEEEeqnarray*}{rCl}
	y_k &=& \mathbb{Q} (D_1=k,S_1=S_0d) ,   \\
	x_k &=& \mathbb{Q} (D_1=k,S_1=S_0u).
\end{IEEEeqnarray*}
It is straightforward to show that $\{x_k\}_{k=0}^n$  and $\{y_k\}_{k=0}^n$ satisfy the assumptions of  Proposition \ref{Prop:QdashwithMarginals} and therefore there exists $\mathbb{Q'} \in \mathcal{M}^e$ with the same 'marginals' as $\mathbb{Q}$ i.e.
	\begin{equation*}
	\mathbb{Q'} (D_1=k,S_1=S_0d) = y_k    \text{ and }    \mathbb{Q'} (D_1=k,S_1=S_0u) = x_k.
	\end{equation*}
It follows that
\begin{equation*}
	\mathbb{E^Q} [D_1(L-S_1)^+] = \mathbb{E^{Q'}} [D_1(L-S_1)^+]
\end{equation*}
and because $\mathcal{M}^e \subset \mathcal{M}$, it follows that (\ref{eq:supEQeqEQdash}) holds.
\end{proof}

We  require that the Combined Market is arbitrage-free. This imposes the following upper bound on the price of the XoL contract:

\begin{Prop}
		A necessary and sufficient condition for the Combined Market    to be arbitrage-free is that the  price $X_0^e$ of the XoL contract must satisfy
		\begin{equation}
			\frac{X_0^e}{n-e}\ < p. \label{eq:XoLNoArb}
		\end{equation}
	\end{Prop}
	
	\begin{proof}
		Suppose that 
		\begin{equation}
			\frac{X_0^e}{n-e} \ge p.\label{eq:XoLArbFalse}
		\end{equation}
		We  create a portfolio which will allow us to make a risk-free profit at time $t=1$. At time $t=0$, we sell an XoL contract with excess $e$ and purchase $(n-e)/n$ GLAs on the $n$ lives. This can be done at a non-positive cost since
		\begin{equation}
			X_0^e-(\frac{n-e}{n})np \ge 0,
		\end{equation}
		from (\ref{eq:XoLArbFalse}). At time $t=1$ there are $D_1$ deaths. We have to pay a claim of $(D_1-e)^+$ and we receive 		$D_1(n-e)/n$ from the GLA assurance. It is straightforward to show that  $(D_1-e)^+\le D_1(n-e)/n$ and for  $0<D_1 < n$ strict inequality holds and so the portfolio is an arbitrage.
		
For the converse, if (\ref{eq:XoLNoArb}) holds then,by Theorem 	\ref{SuperhedgePrices} below, there exists an ACMM $\mathbb{Q}$. Let $\theta \in (0,1)$. Then $\mathbb{Q_\theta} = \theta \mathbb{P}+(1-\theta)\mathbb{Q}$ is an EMM and therefore there is no arbitrage in the market.
\end{proof}
Note that the construction of an EMM $\mathbb{Q}$ in the above proposition shows  that these are multiple EMMs and so the Combined Market is incomplete.

\section{Pricing the NNEG option}

We use the following result in the proof of Theorem \ref{SuperhedgePrices} - namely, the well known (and straightforward to prove) weak duality pricing result
\begin{equation}
	\sup_{\mathbb{Q} \in \mathcal{M}} \mathbb{E}^{\mathbb{Q}}[D_1(L-S_1)^+] \le \inf\{V_0:V_1 \ge D_1(L-S_1)^+\}.			\label{eq:WeakDuality}
\end{equation}

Note that, without loss of generality we may assume that the loan amount $L$ of the lifetime mortgage lies between the upper and lower prices of the property stock i.e.
	\begin{equation}
		S_0d < L \le S_0u,		\label{eq:LbetS0d}
	\end{equation}
because, if $L \le S_0d$, then $D_1(L-S_0d)^+$ is $\mathbb{P}$-a.s. zero and the price of the hedge is zero. Conversely, if $S_0u <L$, then since
\begin{IEEEeqnarray*}{rCl}
	(L-S_1)^+	&=&	L-S_0u+S_0u-S_1 \\
			&=& L-S_0u+(S_0u-S_1)^+,
\end{IEEEeqnarray*}
it follow that 
\begin{IEEEeqnarray*}{rCL}
	\mathbb{E}^\mathbb{Q}[D_1(L-S_1)^+] & = & (L-S_0u)\mathbb{E}^\mathbb{Q}[D_1] + \mathbb{E}^\mathbb{Q} [D_1(S_0u-S_1)^+] \\
	& = & (L-S_0u)np + \mathbb{E}^\mathbb{Q} [D_1(S_0u-S_1)^+]
\end{IEEEeqnarray*}
and the original claim $D_1(L-S_1)^+$ can be superhedged by holding ${(L-S_0u)}$ GLAs  and a superhedging portfolio for the NNEG option with a loan amount of $S_0u$ instead of $L$.

We also assume that 
	\begin{equation}
		L-S_0d = 1																\label{eq:LminusS0}
	\end{equation}
and scale the resulting portfolio amounts if necessary.

\begin{Thm}\label{SuperhedgePrices}
	There is a minimal superhedging portfolio consisting of assets in the Combined Market with price process $\{V_t\}_{t=0,1}$ and a $\mathbb{Q} \in \mathcal{M}^e $ satisfying
	\begin{equation}
		V_0 = \mathbb{E^Q}[D_1(L-S_1)^+].			\label{eq:DualityEquality}
	\end{equation}

\begin{enumerate}
\item[Case 1.] 

If $\frac{X_0^e}{n-e} < q$ and $np \ge eq+X_0^e$, then the minimal superhedging portfolio is constructed from a holding of $e$ property puts and one XoL contract with excess $e$ at a cost of  
\begin{equation}
	V_0 = eq+X_0^e. \label{eq:Case1SH}
\end{equation}
\item[Case 2.]
If $np < eq+X_0^e$, then the minimal superhedging portfolio is constructed from a holding of one GLA contract at a cost of 
\begin{equation}
	V_0=np. 	\label{eq:Case2SH}
\end{equation}
\item[Case 3.]
If $\frac{X_0^e}{n-e} \ge q$, then the minimal superhedging portfolio is constructed from a holding of $n$ property puts at a cost of 
\begin{equation}
	V_0=nq. 	\label{eq:Case3SH}
\end{equation}

\end{enumerate}

\end{Thm}

Note that, in Case 2,  $np \le eq+X_0^e$ implies $\frac{X_0^e}{n-e} < q$, thanks to  the no-arbitrage condition (\ref{eq:XoLNoArb}).

\begin{proof}

It is sufficient to show that the strategies given by (\ref{eq:Case1SH}) - (\ref{eq:Case3SH}) are superhedges i.e. $V_1 \ge D_1(L-S_1)^+$  and show that there is a corresponding ACMM $\mathbb{Q} \in \mathcal{M}^e$ such that 
\begin{equation*}
	 \mathbb{E}^{\mathbb{Q}}[D_1(L-S_1)^+] =V_0.
\end{equation*}
By weak duality (\ref{eq:WeakDuality}), it will follow that $\mathbb{Q}$ is a maximising ACMM and the hedging strategy  a minimising superhedge.

Note that a property put option can be synthesised from cash and a property stock and does not constitute a new asset for hedging purposes. Its price at time $t=0$ will be $q$ because by  (\ref{eq:LbetS0d})  and (\ref{eq:LminusS0}),
\begin{equation}
\mathbb{E}^{\mathbb{Q}}[(L-S_1)^+]=(L-S_0u)^+(1-q)+(L-S_0d)^+q =(L-S_0d)q= q.
\end{equation}

It is easy to show that the strategies are superhedges and we omit the detail. We find the corresponding ACMMs $\mathbb{Q}$:

For Case 1,
 $\frac{X_0^e}{n-e} < q$ and $np \ge eq+X_0^e$:   For $k=0,1, \cdots, n$, let
\begin{IEEEeqnarray*}{rCl}
	y_k & = & 
	\begin{cases}
	\frac{X_0^e}{n-e}       & \text{if } k=n,     \\
	q-\frac{X_0^e}{n-e}    & \text{if } k=e,     \\
	0                                &  \text{otherwise,}
	\end{cases}\\
	\nonumber\\
	x_k & =&
	\begin{cases}
	\frac{1}{e}(np-(eq+X_0^e))                    & \text{if } k=e,\\
	1-q-\frac{1}{e}(np-(eq+X_0^e))             & \text{if } k=0,\\
	0                                                           &  \text{otherwise.}
	\end{cases}
\end{IEEEeqnarray*}
Then it is straightforward to check that $\{x_k\}_{k=0}^n$  and $\{y_k\}_{k=0}^n$, satisfy the conditions of Proposition \ref{Prop:QdashwithMarginals} and therefore there exists a $\mathbb{Q} \in \mathcal{M}^e$ such that
\begin{equation*}
	x_k = \mathbb{Q} (D_1=k,S_1=S_0u)     \text{ and }    y_k = \mathbb{Q} (D_1=k,S_1=S_0d),
\end{equation*}
for $k=0,1,\cdots, n$. The expectation of the claim (\ref{eq:NNEGPayout}) under the ACMM $\mathbb{Q}$ satisfies:
\begin{IEEEeqnarray*}{rCl}
	\mathbb{E^Q}[D_1(L-S_1)^+] & = & \sum_{k=0}^n k(L-S_0u)^+x_k +  \sum_{k=0}^n k(L-S_0d)^+y_k \\
	& = & e \bigg{(}q-\frac{X_0^e}{n-e} \bigg{)} + n\frac{X_0^e}{n-e} \\
	& = & eq + X_0^e.
\end{IEEEeqnarray*}
\\\\
For Case 2,
$np < eq+X_0^e$: For $k=0,1, \cdots, n$, let
\begin{IEEEeqnarray*}{rCl}
	y_k  & = & 
	\begin{cases}
	\frac{X_0^e}{n-e}					& \text{if } k=n,\\
	\frac{n}{e}(p-\frac{X_0^e}{n-e})			& \text{if } k=e,\\
	q-\frac{1}{e}(np-X_0^e)				& \text{if } k=0, \\
	0								 & \text{otherwise,}
	\end{cases}\\
	\nonumber\\
	x_k  & =&
		\begin{cases}
		1-q							& \text{if } k=0,	\\
		0							&  \text{otherwise.}
		\end{cases}
\end{IEEEeqnarray*}
Then  $\{x_k\}_{k=0}^n$  and $\{y_k\}_{k=0}^n$  satisfy the conditions of Proposition \ref{Prop:QdashwithMarginals}. Note that $y_0 \ge 0$ because $q-\frac{1}{e}(np-X_0^e)\ge q-\frac{1}{e}(eq)=0$. There exists a $\mathbb{Q} \in \mathcal{M}^e$ such that
\begin{equation*}
	x_k = \mathbb{Q} (D_1=k,S_1=S_0u)     \text{ and }    y_k = \mathbb{Q} (D_1=k,S_1=S_0d),
\end{equation*}
for $k=0,1,\cdots, n$. The expectation of the claim (\ref{eq:NNEGPayout}) under the ACMM $\mathbb{Q}$ satisfies:

\begin{IEEEeqnarray*}{rCl}
	\mathbb{E^Q}[D_1(L-S_1)^+] & = & \sum_{k=0}^n k(L-S_0u)^+x_k +  \sum_{k=0}^n k(L-S_0d)^+y_k \\
	& = & e \bigg{(} \frac{1}{e} \bigg{\{} np - \frac{n}{n-e}X_0^e \bigg{\}} \bigg{)} + n \bigg{(} \frac{X_0^e}{n-e} \bigg{)}\\
	& = & np.
\end{IEEEeqnarray*}

For Case 3,
$\frac{X_0^e}{n-e} \ge q$: For $k=0,1, \cdots, n$, let
\begin{IEEEeqnarray*}{rCl}
	y_k =  & = & 
	\begin{cases}
	q						& \text{if } k=n,			\\
	0						& \text{otherwise,}
	\end{cases}\\
	\nonumber\\
	x_k =  & =&
	\begin{cases}
	\frac{X_0^e}{n-e}-q					& \text{if } k=n,		\\
	\frac{n}{e}(p-\frac{X_0^e}{n-e})			& \text{if } k=e,		\\
	1-\frac{1}{e}(np-X_0^e)				& \text{if } k=0,		 \\
	0								&  \text{otherwise.}
	\end{cases}
\end{IEEEeqnarray*}
Then  $\{x_k\}_{k=0}^n$  and $\{y_k\}_{k=0}^n$ satisfy the conditions of Proposition \ref{Prop:QdashwithMarginals} and there exists a $\mathbb{Q} \in \mathcal{M}^e$ such that
\begin{equation*}
	x_k = \mathbb{Q} (D_1=k,S_1=S_0u)     \text{ and }    y_k = \mathbb{Q} (D_1=k,S_1=S_0d),
\end{equation*}
for $k=0,1,\cdots, n$. The expectation of the claim (\ref{eq:NNEGPayout}) under the ACMM $\mathbb{Q}$ satisfies $\mathbb{E^Q}[D_1(L-S_1)^+]  = nq$.
\end{proof}

\section{Varying loan amounts and death probabilities}
\subsection{Varying loan amounts}\label{VaryingLoans}
So far we have assumed that the loan amount $L$ is the same for all ERM contracts. In practice such an assumption is often not a serious issue because we can value and hedge contracts with similar-sized loan amounts together. 

In this section, we provide some superhedging strategies for varying loan amounts. We shall see that, by using suitably modified XoL and GLA assets, the strategies we found previously are still superhedges, and thus give an upper bound to the minimal hedging cost. However, only under certain circumstances will these strategies be the cheapest possible.

Note that the property \lq\lq price'' we use is really an index and the general sum at risk will be of the form
$P_i(L_i-S_0d)$ where $P_i$ is the initial value of the property and $L_i$ is the loan to value ratio [LTVR].

Assume that the LTVR $L_i$ for the $i^{th}$ life satisfies $S_0d<L_i<S_0u$, let $\alpha_i =P_i( L_i - S_0d)$, for $i=1,2,\cdots, n$ and let $\Sigma = \sum_{i=1}^n \alpha_i$.

Instead of the life-symmetrical payout (\ref{eq:PropPutPayout}), the NNEG payout now has the form
\begin{equation}
	(1-\omega_0) \sum_{i=1}^n \alpha_i \omega_i, \qquad \text{for } \omega = (\omega_0, \cdots \omega_n) \in \Omega.		\label{eq:NonSymPayout}
\end{equation} 

We modify the original XoL contract to have the payout
\begin{equation*}
	\bigg(\sum_{i=1}^n \alpha_i \omega_i - e \bigg)^+, \qquad \text{for } \omega = (\omega_0, \cdots \omega_n) \in \Omega,
\end{equation*}
where $e = \sum_{k=1}^m \alpha_{i_k}$, for a specified subsequence $(\alpha_{i_1}, \cdots, \alpha_{i_m})$ of $(\alpha_1, \cdots, \alpha_n)$.

The GLA contract will now have the payout 
\begin{equation}
	\sum_{i=1}^n \alpha_i \omega_i, \qquad \text{for } \omega = (\omega_0, \cdots \omega_n) \in \Omega.		\label{eq:ModifiedGLA}
\end{equation}

The set $\mathcal{M}$ of ACMMs  changes to $\mathcal{N}$ consisting of all probability measures $\mathbb{Q}$ satisfying (\ref{eq:ACMMStock}), (\ref{eq:ACMMSLA}) and 
	\begin{equation}
		X_0^e=\sum_{\omega \in \Omega}	\mathbb{Q}(\omega) \sum_{i=1}^n (\alpha_i \omega_i - e)^+ .				 \label{eq:ACMMXoLNew}
	\end{equation}

Unlike the previous constant loan amount case, we can no longer replace the single-life pricing constraint by the easier-to-handle GLA constraint. 

The XoL no-arbitrage condition (\ref{eq:XoLNoArb}) becomes
\begin{equation}
			\frac{X_0^e}{\Sigma-e}\ < p. \label{eq:XoLNoArbNew}
\end{equation}

The three superhedging strategies from Theorem \ref{SuperhedgePrices} become 
\begin{enumerate}
	\item[SH1:] \label{itm:one} 		Hold an XoL contract with excess $e = \sum_{k=1}^m \alpha_{i_k}$ and a basket of $m$ property puts with strikes $L_{i_1}, \cdots, L_{i_m}$. Note that the setup cost of this strategy is still $eq+X_0^e$. 
	\item[SH2:] \label{itm:two} 	Hold a GLA contract with payout (\ref{eq:ModifiedGLA}) at a cost of $\Sigma$p.
	\item[SH3:] \label{itm:three}		Hold a basket of $n$ property puts with strikes  $L_1, \cdots, L_n$ at a cost of $\Sigma q$.
\end{enumerate}


It is easy to show that the above three strategies are superhedges for the claim (\ref{eq:NonSymPayout}) but it is necessary to check whether there are cheaper strategies using linear programming or the following version of Theorem \ref{SuperhedgePrices} :

\begin{Prop}\label{SuperhedgePricesNew}
	Let $\omega^{e} \in \Omega$ satisfy $e = \sum_1^n \alpha_i \omega^{e}_i$ and $\omega^{e}_0 = 0$. Let $\tilde{\omega}^e \in \Omega$ satisfy $\tilde{\omega}^e_i = 1-\omega^{e}_i, i=0,1,\cdots,n$. Assume that $\sum_1^n \alpha_i \omega_i^{e} > \sum_1^n \alpha_i \tilde{\omega}_i^e$ and that $1-2p+\frac{X_0^e}{\Sigma-e} \ge 0$.
	Then there is a minimal superhedging portfolio consisting of assets in the Combined Market with price process $\{V_t\}_{t=0,1}$ and a $\mathbb{Q} \in \mathcal{N} $ satisfying
	\begin{equation}
		V_0 = \mathbb{E^Q} \bigg[\sum_{\omega \in \Omega}  (1-\omega_0) \sum_{i=1}^n \alpha_i \omega_i \mathbb{Q} (\omega) \bigg].
	\end{equation}

\begin{enumerate}
\item[Case 1.]

If $\frac{X_0^e}{\Sigma-e} < q$ and $p \ge q$, then the minimal superhedging portfolio is constructed from a holding of $m$ property puts with strike $L_{i_k}, k=1,\cdots,m$ and one XoL contract with excess $e$ at a cost of  
\begin{equation}
	V_0 = eq+X_0^e. \label{eq:Case1SHNew}
\end{equation}
\item[Case 2.]
If $p < q$, then the cheapest superhedge might not be one of the three superhedges SH1-SH3 and we need to use  linear optimisation to obtain the cheapest strategy.
\item[Case 3.]
If $\frac{X_0^e}{\Sigma-e} \ge q$, then the minimal superhedging portfolio is constructed from a holding of $n$ property puts with strike $L_i, i=1,\cdots,n$ at a cost of 
\begin{equation}
	V_0=\Sigma q. 	\label{eq:Case3SHNew}
\end{equation}

\end{enumerate}

\end{Prop}

\begin{proof}

It is sufficient to show that the strategies  (\ref{eq:Case1SHNew}) and (\ref{eq:Case3SHNew}) are superhedges i.e. $V_1(\omega)  \ge (1-\omega_0) \sum_{i=1}^n \alpha_i \omega_i$  and show that there is a corresponding ACMM $\mathbb{Q} \in \mathcal{N}$ such that 
\begin{equation*}
	 \mathbb{E}^{\mathbb{Q}} \bigg[(1-\omega_0) \sum_{i=1}^n \alpha_i \omega_i \bigg] =V_0.
\end{equation*}

We find the corresponding ACMMs $\mathbb{Q}$:
for Case 1, define $\mathbb{Q} \in \mathcal{N}$ as follows:
\begin{itemize}
	\item 	$\mathbb{Q}((0,1,1,\cdots, 1)) = \frac{X_0^e}{\Sigma-e} $
	\item 	$\mathbb{Q}(\omega^{e}) = q-\frac{X_0^e}{\Sigma-e}$
	\item		$\mathbb{Q}((1,\omega^{e}_1, \cdots, \omega^{e}_n)) = p-q$
	\item 	$\mathbb{Q}(\tilde{\omega}^e) = p-\frac{X_0^e}{\Sigma-e}$
	\item 	$\mathbb{Q}((1,0,0,\cdots,0)) = 1-2p+\frac{X_0^e}{\Sigma - e}$
	\item         for all other $\omega \in \Omega$ set $\mathbb{Q}(\omega) = 0$.
\end{itemize}
We need to show that $\mathbb{Q}$ satisfies the pricing constraints (\ref{eq:ACMMStock}), (\ref{eq:ACMMSLA}) and  (\ref{eq:ACMMXoLNew}). We show (\ref{eq:ACMMSLA}):

\begin{IEEEeqnarray*}{rCl}
	\mathbb{Q} (i^{th} \text{ life dies}) 	&=& 		\sum_{\omega \in \Omega, \omega_i = 1} \mathbb{Q} (\omega)		\\
								&= &		\begin{cases}
											\frac{X_0^e}{\Sigma-e} + \big(q-\frac{X_0^e}{\Sigma-e} \big)  + (p-q)				& \text{if } \omega_i^e=1,			\\
											\frac{X_0^e}{\Sigma-e} + \big(p-\frac{X_0^e}{\Sigma-e} \big)					& \text{otherwise,}
								     		 \end{cases}	\\
								 &=&		p.
\end{IEEEeqnarray*}

The expectation of (\ref{eq:NonSymPayout}) under $\mathbb{Q}$ satisfies
\begin{IEEEeqnarray}{rCl}
	\mathbb{E}^\mathbb{Q} \bigg[ (1-\omega_0) \sum_{i=1}^n \alpha_i \omega_i \bigg] &=& \Sigma \frac{X_0^e}{\Sigma-e} + e\big(q - \frac{X_0^e}{\Sigma-e} \big)  \\
						&=& eq + X_0^e.
\end{IEEEeqnarray}

For Case 2, we provide the following counter-example to show that the superhedges SH1-SH3 are not necessarily the cheapest:

Let the number of lives be three with loan amounts $L_1=70, L_2=80$ and $L_3=90$. Let the probability of death $p=0.45$. The XoL contract has an excess  $e=70$ and costs $X_0^e=1.822$.

Let the property stock have initial price $S_0=100$ and price $S_1=160$ in an "up" scenario and $S_1=50$ in a "down" scenario. The risk-neutral probability that $S_1=50$ is $q=0.5454$.
The shortfall in a "down" scenario on each of the ERM policies is $\alpha_1=20,\alpha_2=30,\alpha_3=40$.

The setup costs of each of the three superhedges is 
\begin{itemize}
	\item		cost of SH1 $=eq + X_0^e = 40.00$
	\item		cost of SH2 $= (\alpha_1+\alpha_2+\alpha_3)p=40.5$
	\item		cost of SH3 $= (\alpha_1+\alpha_2+\alpha_3)q=49.09$.
\end{itemize}

But a cheapest superhedge (found by linear programming) is a strategy consisting of one property put option with strike 90, one XoL contract, $\alpha_2/3$ of a single life assurance on life 2 and $\alpha_3/2$ of a single life assurance on life 3 at a cost of $37.14$.

For Case 3, let
\begin{itemize}
	\item 	$\mathbb{Q}((0,1,1,\cdots, 1)) = q $
	\item 	$\mathbb{Q}((1,1,1,\cdots, 1)) = \frac{X_0^e}{\Sigma-e} - q$
	\item		$\mathbb{Q}((1,\omega^{e}_1, \cdots, \omega^{e}_n)) = p-\frac{X_0^e}{\Sigma-e}$
	\item 	$\mathbb{Q}(\tilde{\omega}^e) = p-\frac{X_0^e}{\Sigma-e}$
	\item 	$\mathbb{Q}((1,0,0,\cdots,0)) = 1-2p+\frac{X_0^e}{\Sigma - e}$
	\item         for all other $\omega \in \Omega$ set $\mathbb{Q}(\omega) = 0$.
\end{itemize}
Then $\mathbb{Q} \in \mathcal{N}$ and $\mathbb{E}^\mathbb{Q} \big[ (1-\omega_0) \sum_{i=1}^n \alpha_i \omega_i \big] = \Sigma q$.

\end{proof}

\subsection{Varying death probabilities}

In practice, lives with similar ages can be grouped together for pricing and hedging. However, in this section, we give a brief outline on how to generalise the hedging framework of Section \ref{VaryingLoans} to allow the probability of death to vary by life within a portfolio. 

We denote the probability of death of the $i^\text{th}$  life by $p_i,i=1,2,\cdots, n$. The set $\mathcal{N}$ of ACMMs changes to the set of all probability measures $\mathbb{Q}$ satisfying (\ref{eq:ACMMStock}), (\ref{eq:ACMMXoLNew}) and 

\begin{equation}
	\mathbb{Q}(L_1^{(i)}=1) =p_i, \qquad i=1,\cdots, n. \label{eq:ACMMSLANew}
\end{equation}

The XoL no-arbitrage condition (\ref{eq:XoLNoArbNew}) becomes
\begin{equation}
			\frac{X_0^e}{\Sigma-e}\ < \frac{1}{\Sigma} \sum_{i=1}^n \alpha_i p_i. \label{eq:XoLNoArbNewNew}
\end{equation}

The strategies SH1-SH3 of Section \ref{VaryingLoans} remain superhedges but the price of strategy SH2 changes from $\Sigma p$ to $\sum_1^n \alpha_i p_i$.

The two pricing constraints (\ref{eq:ACMMXoLNew}) and (\ref{eq:ACMMSLANew}) make it more difficult to find an equivalent result to Theorem \ref{SuperhedgePrices} and, in general, it will be necessary to use linear programming to find a cheapest superhedge.

\section{Large deviations and tracking errors}\label{section:LargeDevs}

In what follows, we assume that the XoL contract on $n$ lives has the excess given by (\ref{eq:e(n)}) and it will be convenient to define $a=p(1+\epsilon)$ so that $e(n)=na$.

Define  the rate function $I_p$ as follows:
\begin{equation}
	I_p(a)=a\log \frac{a}{p} + (1-a)\log \frac{1-a}{1-p},\qquad p,a \in (0,1).
\end{equation}
Recall the following well known large deviations estimate:

\begin{Lem}\label{LDPresult} (see \cite{grimmett2001probability})
	Let $\{ X_i \}_{i=1}^{\infty}$ be a sequence of independent, identically-distributed Bernoulli random variables with success probability $p \in (0,1)$. 
	
	Let $S_n = \sum_{i=1}^n X_i$. Then  for $a\in (p,1)$,  $I_p(a)>0$ and 
	\begin{equation}
		\mathbb{P}(S_n \ge na) \le e^{-nI_p(a)}, \qquad n=1,2,\cdots.
	\end{equation}

\end{Lem}

 We introduce the reinsurer's pricing measure $\mathbb{P}^b$ to allow for the mortality loading that the reinsurer may use in their pricing. $\mathbb{P}^b$ will not affect the excess $e(n)=np(1+\epsilon)$ only the probability of death used in pricing.
\begin{Def}
Define $\mathbb{P}^b$ by the property that the lives are independent and Bernoulli-distributed, with probability of death $b:=p(1+\eta)$, $0\leq\eta<\epsilon$.
\end{Def}

We denote the integer part of $x$ by $\lfloor x \rfloor$.

\begin{Prop}\label{XoL_LDP}

Assume that the reinsurer prices the XoL contract with a margin for risk and profit using the pricing measure $\mathbb{P}^b$ with $e(n)$ given by (\ref{eq:e(n)}). Then
\begin{equation*}
	X_0^e(n;b) \le (n-\lfloor e(n) \rfloor) e^{-nI_b(a)}.
\end{equation*}
and $X_0^e(n;b)$ tends to zero exponentially fast as $n \rightarrow \infty$.

\begin{proof}
Since $D_1(n)$ is a sum of $n$ independent Bernoulli random variables with probability of success(death) $b$, Lemma \ref{LDPresult} applies.  We have
 \begin{IEEEeqnarray}{rCl}
 	X_0^e(n;b) &=& \sum_{k>e(n)} (k-e(n)) \mathbb{P}^b(D_1 = k) \nonumber \\
			&\le& \sum_{k>e(n)} (k-\lfloor e(n) \rfloor) \mathbb{P}^b(D_1 = k) \nonumber \\
			&=& \sum_{k=1}^{n-\lfloor e(n) \rfloor} k \mathbb{P}^b(D_1 = k+\lfloor e(n) \rfloor) \nonumber \\
			&=& \sum_{k=1}^{n-\lfloor e(n) \rfloor} \mathbb{P}^b(D_1 \ge k+\lfloor e(n) \rfloor) \nonumber \\
			&\le& (n-\lfloor e(n) \rfloor) \mathbb{P}^b(D_1 \ge e(n))				\nonumber \\
			&\le& (n-\lfloor e(n) \rfloor) e^{-nI_b(a)},\qquad \text{by Lemma } \ref{LDPresult}.
 \end{IEEEeqnarray}

\end{proof}
\end{Prop}

The previous proposition showed that $\frac{1}{n}X_0^e(n;b) \rightarrow 0$ extremely fast as $n \rightarrow \infty$. The next proposition shows that if
\begin{equation}
	\frac{1}{n}X_0^e(n;b) < \min \big{(} p(1-q(1+\epsilon)),q(1-p(1+\epsilon)  \big{)}, \label{eq:n0}
\end{equation}
then the minimal superhedge consists of $e(n)$ property puts and an XoL contract so that the average cost per policyholder of the superhedge is
\begin{equation*}
	\frac{1}{n}V_0(n) = p(1+\epsilon)q +\frac{1}{n}X_0^e(n;b).
\end{equation*}
This means that the average cost of the superhedge tends exponentially fast to the asymptotic average cost $qp(1+\epsilon)$ which, for small enough $\epsilon$, is less that the average cost $p$ of holding only a GLA or $q$ -  the average cost of holding  $n$ property puts.

\begin{Prop}\label{AsymptoticHedge}
Assume that the reinsurer prices the XoL contract using the  measure $\mathbb{P}^b$ with $e(n)$ given by (\ref{eq:e(n)}) and that $\epsilon$ is sufficiently small so that $\max(p,q) < \frac{1}{1+\epsilon}$. Then, for $n$ large, the minimal superhedge consists of $e(n)$ property puts and an XoL reinsurance contract. The tracking error or difference between the average cost of the superhedge $\frac{1}{n}V_0(n)$ and the average cost $pq$ of the asymptotic portfolio is given by
\begin{equation*}
	\frac{1}{n}V_0(n)-pq = \epsilon pq + \frac{1}{n}X_0^e(n;b).
\end{equation*}

To avoid arbitrage, $X_0^e(n;b)$ needs to satisfy the no-arbitrage condition (\ref{eq:XoLNoArb}) as before.

\end{Prop}
\begin{proof}
Since $X_0^e(n;b) \rightarrow 0$ exponentially fast, there is an $N_0 \in \mathbb{N}$ such that for all $n \ge N_0$, $X_0^e(n;b)$ satisfies (\ref{eq:n0}).
Note that 
\begin{equation}
	\frac{1}{n}X_0^e(n;b) < q(1-p(1+\epsilon))  \iff \frac{X_0^e(n;b)}{n-e(n)} <q \label{eq:XCond1}
\end{equation}
and
\begin{equation}
	\frac{1}{n}X_0^e(n;b) < p(1-q(1+\epsilon))  \iff X_0^e(n;b) +e(n)q < np. \label{eq:XCond2}
\end{equation}
By Theorem \ref{SuperhedgePrices}, the minimal superhedging portfolio consists of $e(n)$ property puts costing $q$ each and an XoL contract costing $X_0^e(n;b)$. The rest of the proposition follows because $V_0(n) = e(n)q +X_0^e(n;b)$.

\end{proof}
Note that Proposition \ref{AsymptoticHedge} implies that the asymptotic average hedging cost is unaffected by  the mortality loading $\eta$ that the reinsurers adds to the probability of death $p$ so long as  $\eta \le \epsilon$.

\bibliographystyle{plain}
\bibliography{DiscreteERMsub}

\begin{bibdiv}
\begin{biblist}

\bib{CP48/16}{techreport}{
      author={(2016), Prudential Regulatory~Authority},
       title={\text{CP}48/16: Matching adjustment - illiquid unrated assets and
  equity release mortgages},
 institution={Bank of England},
        date={2016},
}

\bib{SS2/17}{techreport}{
      author={Authority(2020), Prudential~Regulatory},
       title={\text{SS}2/17: Illiquid unrated assets},
 institution={Bank of England},
        date={2020},
}

\bib{MathArbit2006}{book}{
      author={Delbaen, F.},
      author={Schachermayer, W.},
       title={The mathematics of arbitrage},
   publisher={Springer Finance},
        date={2006},
}

\bib{dowd2018asleep}{article}{
      author={Dowd, K.},
       title={Asleep at the wheel: The prudential regulation authority and the
  equity release sector},
        date={2018},
     journal={Adam Smith Institute},
}

\bib{dowd2019valuation}{article}{
      author={Dowd, K.},
      author={Buckner, D.},
      author={Blake, D.},
      author={Fry, J.},
       title={The valuation of no-negative equity guarantees and equity release
  mortgages},
        date={2019},
     journal={Economics Letters},
      volume={184},
       pages={108669},
}

\bib{follmer2011stochastic}{book}{
      author={F{\"o}llmer, H.},
      author={Schied, A.},
       title={Stochastic finance: an introduction in discrete time},
   publisher={Walter de Gruyter},
        date={2011},
}

\bib{grimmett2001probability}{book}{
      author={Grimmett, G.},
      author={Stirzaker, D.},
       title={Probability and random processes},
   publisher={OUP Oxford},
        date={2001},
        ISBN={9780198572220},
}

\bib{jacka2015coupling}{article}{
      author={Jacka, SD.},
      author={others},
       title={Coupling and tracking of regime-switching martingales},
        date={2015},
     journal={Electronic Journal of Probability},
      volume={20},
}

\bib{IrishERM}{techreport}{
      author={Jeffery, T.},
      author={Smith, AD.},
       title={Equity release mortgages: Irish \& \text{UK} experience},
 institution={Society of Actuaries in Ireland},
        date={2019},
}

\end{biblist}
\end{bibdiv}

\end{document}